\documentclass[11pt,letterpaper,fleqn]{article}
\usepackage{mathrsfs}
\usepackage[latin1]{inputenc}
\usepackage{amsmath}
\usepackage{amsthm}
\usepackage{amsfonts}
\usepackage{amssymb}
\usepackage{graphicx}
\usepackage{breqn}
\usepackage{framed}
\usepackage[numbers]{natbib}

\newtheorem{claim}{Claim}

\newcommand{\dhc}{\textsc{Directed-Happy-Cat}}
\newcommand{\uhc}{\textsc{Undirected-Happy-Cat}}
\newcommand{\smcvp}{\textsc{SMCVP}}

\newcommand{\gcat}{G_C}
\newcommand{\gmouse}{G_M}

\newcommand{\gcmh}{\langle G, c, m, h \rangle}

\title{Undirected Cat-and-Mouse is P-complete}
\author{Arefin Huq\\
Georgia Institute of Technology\\
arefin@gatech.edu}
\date{}
\begin{document}
\maketitle

\begin{abstract}
Cat-and-mouse is a two-player game on a finite graph. Chandra and Stockmeyer showed cat-and-mouse is P-complete on directed graphs. We show cat-and-mouse is  P-complete on undirected graphs.

To our knowledge, no proof of the directed case was ever published. To fill this gap we give a proof for directed graphs and extend it to undirected graphs. The proof is a reduction from a variant of the circuit value problem.
\end{abstract}

\section{Introduction}
Cat-and-mouse is a two-player game on a finite directed or undirected graph. The directed version was presented in Chandra and Stockmeyer's seminal paper on alternation \citep{DBLP:conf/focs/ChandraS76} as an example of a game that is P-complete.\footnote{The actual statement was that cat-and-mouse is log-complete for alternating logspace, which by a central result of the same paper is equal to polynomial time.} The undirected version appears in an exercise from Sipser's classic textbook on the theory of computation \cite{sipser-book}.

In both versions, the players Cat and Mouse alternate traversing edges from the node they are on to an adjacent node. Cat's goal is to catch Mouse by occupying the same node. Mouse's goal is to get to a designated node (the Hole) before that happens. 

In this paper we show that cat-and-mouse on undirected graphs is P-complete, extending Chandra and Stockmeyer's result for the directed case. The proof is a reduction from a variant of the circuit value problem.

\subsection{Missing Proof of the Directed Case}

To our knowledge, no proof of the directed case has ever been published. Chandra and Stockmeyer's result for directed cat-and-mouse appears as Theorem 4.2 of the conference version of their alternation paper \citep{DBLP:conf/focs/ChandraS76} with a placeholder stating, ``The proof of Theorem 4.2 will appear in the final version of this paper.'' However, the final version \citep{DBLP:journals/jacm/ChandraKS81} makes no mention of cat-and-mouse or Theorem~4.2.

The closest thing to a published proof we found is Section~A.11.2 of the book on parallel computation by Greenlaw, Hoover and Ruzzo \citep{p-complete-compendium}, which proposes a reduction from a logspace alternating Turing machine and describes how to handle existential configurations. The authors cite a personal communication from Stockmeyer, which leads us to believe the original unpublished proof took this form.

To fill this gap in the literature we prove the directed case and then generalize to the undirected case. Our reduction is from a circuit rather than a logspace alternating Turing machine.

\section{Definitions}

\subsection{The cat-and-mouse game}
\label{sec:sipser-cat-and-mouse}
The following is Problem 8.15 of Sipser \cite{sipser-book}:
\begin{framed}
{\em The cat-and-mouse game is played by two players, ``Cat'' and ``Mouse,'' on an arbitrary undirected graph. At a given point each player occupies a node of the graph. The players take turns moving to a node adjacent to the one that they currently occupy. A special node of the graph is called ``Hole.'' Cat wins if the two players ever occupy the same node. Mouse wins if it reaches the Hole before the preceding happens. The game is a draw if a situation repeats (i.e., the two players simultaneously occupy positions that they simultaneously occupied previously and it is the same player's turn to move).}
\begin{align*}
&\textsc{Happy-Cat} = \\
&\qquad\{ \gcmh \,\vert\, G,c,m,h, \text{are respectively a graph, and } \\
&\qquad\qquad\text{positions of the Cat, Mouse, and Hole, such that} \\
&\qquad\qquad\text{Cat has a winning strategy if Cat moves first} \}.
\end{align*}
{\em Show that \textsc{Happy-Cat} is in P.}
\end{framed}

We use this definition of cat-and-mouse, though for clarity we refer to \textsc{Happy-Cat} as \uhc{}. We define \dhc{} identically except that $G$ is a directed graph.

\subsubsection{Differences with Chandra and Stockmeyer}
The game Chandra and Stockmeyer \citep{DBLP:conf/focs/ChandraS76} present differs from \dhc{} in ways that do not affect our argument. Here are the differences along with explanations of why they do not affect our construction, assuming optimal play:
\begin{enumerate}
\item \textit{Either player may pass at any time.}
In our construction passing never helps.

\item \textit{Cat may not occupy Hole.}
In our construction, if Cat can beat Mouse to Hole, Cat can win before then.

\item \textit{Mouse moves first.}
Cat's first move is extraneous -- see Section~\ref{sec:opening-move}. 

\item \textit{The language is the set of instances where Mouse wins.}
In our construction draws can't happen under optimal play, so Mouse wins exactly if Cat doesn't.
\end{enumerate}

\subsection{Synchronous Monotone Circuit Value Problem}

A Boolean circuit is \emph{synchronous} if all paths from input to output have the same length, and is \emph{monotone} if it only has AND and OR gates. The synchronous monotone circuit value problem (here denoted \smcvp{}) asks whether a given encoding of a synchronous monotone Boolean circuit evaluates to true on a given input. Greenlaw, Hoover and Ruzzo \citep{p-complete-compendium} showed this problem is P-complete even when all gates have fan-in and fan-out two.

\section{The construction}

Recall that a language is \emph{P-complete} if it is in P and every language in P is logspace reducible to it. Showing that cat-and-mouse is in P is left in the form given by Sipser: as an exercise for the reader. (The proof for the directed case is similar to the undirected case.) To show that cat-and-mouse is P-hard it suffices to give a reduction from \smcvp{}.

Given a synchronous monotone circuit $C$ and input assignment $x$, we show how to construct an instance of the directed cat-and-mouse game such that Cat has a winning strategy exactly when the circuit is \emph{false}. We then show how to convert this to an undirected game with the same property. The rough idea is that the game graph encodes two parallel copies of the circuit. Mouse races from the output to an input, chased by Cat. If the circuit is true Mouse will have a path that evades Cat, otherwise Cat can intercept any path Mouse takes.

A schematic of the construction for both cases is depicted in Figure~\ref{fig:overview}. Mouse moves from $m$ toward $h$ in the Mouse subgraph, and Cat (starting from $c$) mirrors Mouse in the Cat subgraph. If either player deviates from this strategy they lose. If both players follow the strategy then Mouse makes it to Hole exactly when $C(x) = 1$.

We should point out that our construction for the directed case is more complicated than necessary. The payoff for the added complexity is that the transformation to the undirected case is straightforward.

\subsection{Directed case}

We construct a directed game graph $G_D$ from $C$ and $x$ as follows:

\begin{enumerate}
\itemsep 0pt
\parskip 0pt
\item Let $G$ denote the DAG corresponding to $C$.
\item Construct $G'$ by replacing each gate node in $G$ with the directed gate gadget (Figure~\ref{fig:gate-gadget}).
\item Create the Cat and Mouse subgraphs, $\gcat$ and $\gmouse$, as copies of $G'$.
\item \label{item:start-nodes} Add a node $c$ and connect from it to the output of $\gcat$. Label the output node of $\gmouse$ as $m$. These are the starting positions of Cat and Mouse, respectively.
\item Add a node $h$ (the Hole). Connect to it from the nodes in $\gcat$ and $\gmouse$ that correspond to true inputs of $x$.
\item \label{item:dead-end} Add a node $d$ (not shown). Connect to it from the nodes in $\gmouse$ that correspond to false inputs and from all nodes in $\gcat$ that correspond to inputs.
\item \label{item:threat-edges} Add threat edges (Section~\ref{sec:and-gates}) to enforce AND-gate semantics.
\item \label{item:escape-routes} Add escape routes (Section~\ref{sec:escape-routes}) to force Cat to mirror Mouse.
\end{enumerate}

\subsection{The construction - undirected case}
From $G_D$ we construct an undirected game graph $G_U$ as follows:
\begin{enumerate}
\itemsep 0pt
\parskip 0pt
\item Replace each edge in $G_D$ with an undirected edge.
\item \label{item:guard-edges} Add guard edges (Section~\ref{sec:guard-edges}) to force Mouse forward.
\end{enumerate}

\begin{figure}[ht]
\centering
\includegraphics{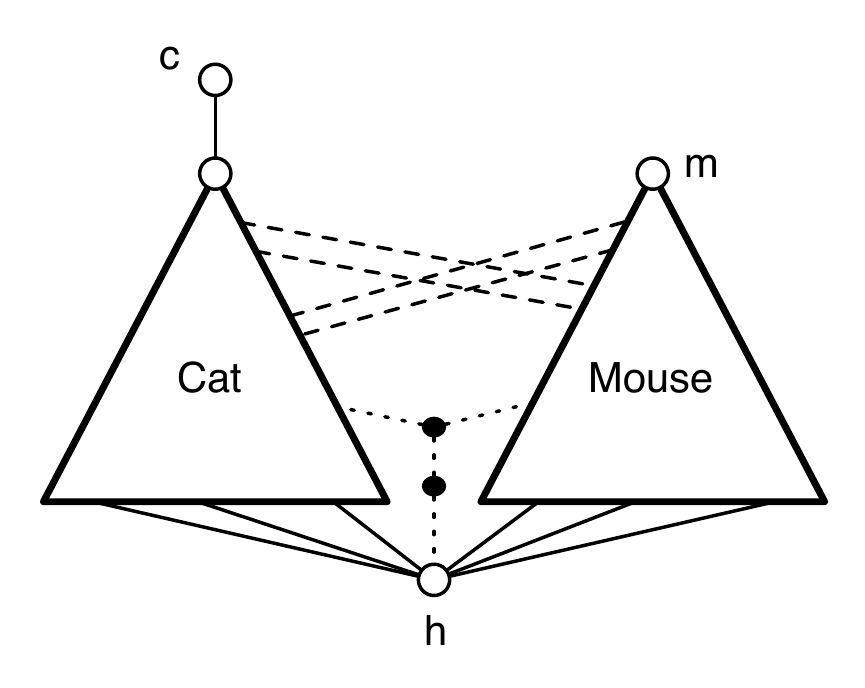}
\caption{Overview of the construction, showing the Cat and Mouse subgraphs, the starting positions $c$ and $m$, and the Hole $h$. In the directed case, all directed edges are from higher levels to lower levels. The dashed lines represent the threat edges (and guard edges for the undirected case) between the subgraphs. One escape route is shown by the dotted lines and filled in nodes connected to $h$.}
\label{fig:overview}
\end{figure}

\begin{figure}[ht]
\centering
\includegraphics{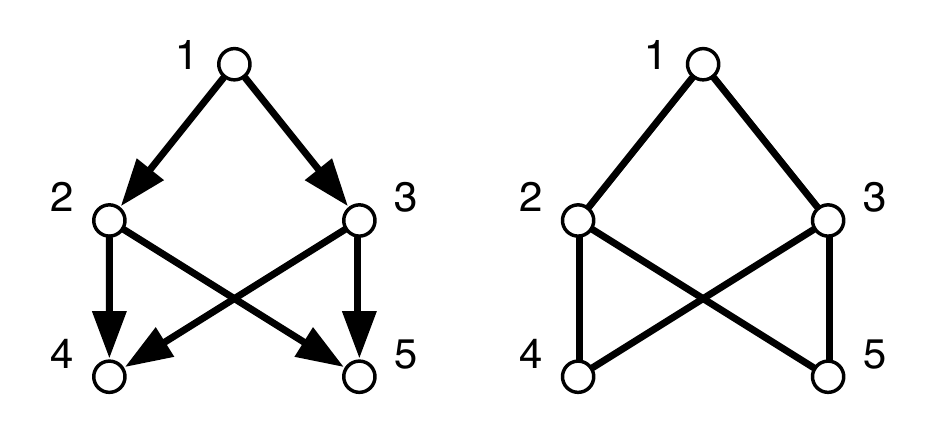}
\caption{The gate gadget, directed (left) and undirected (right).}
\label{fig:gate-gadget}
\end{figure}

\subsection{Mirroring}
The subgraphs $\gcat$ and $\gmouse$ are identical, so there is a one-to-one correspondence between their nodes. If $u$ is a node in $\gmouse$ then let $\gcat(u)$ denote the corresponding node in $\gcat$. We say that Cat \emph{mirrors} Mouse if Cat moves to $\gcat(u)$ whenever Mouse moves to a node $u$ in $\gmouse$.  (There is one exception: in the gadget for an AND gate, shown in Figure~\ref{fig:and-basic}, when Mouse moves to M2 or M3, Cat is free to move to either C2 or C3; this is still considered mirroring.)

\subsection{The opening move}
\label{sec:opening-move}
Cat is required to move first but we want Cat to follow Mouse. To fix this, in step \ref{item:start-nodes} of the construction we constrain Cat's first move to be from $c$ to the output node of $\gcat$. After this the players are in the same layer with Mouse to move.

\subsection{Encoding OR-gates}

For each OR-gate there is an instance of the gate gadget (Figure~\ref{fig:gate-gadget}) in each of the Cat and Mouse subgraphs. Node 1 corresponds to the output of the gate and nodes 4 and 5 correspond to the two inputs.

Typically, Mouse enters the gadget through Node 1 and exits through Node 4 or 5, on a true branch if possible. Mouse can exit on a true branch iff the gate is true. Typically, Cat will mirror Mouse.

Only nodes 1-3 are really needed to encode OR-gate semantics. The rest are there so that the gadget is the same for OR and AND gates.

\subsection{Encoding AND-gates}
\label{sec:and-gates}

\begin{figure}[ht]
\centering
\includegraphics{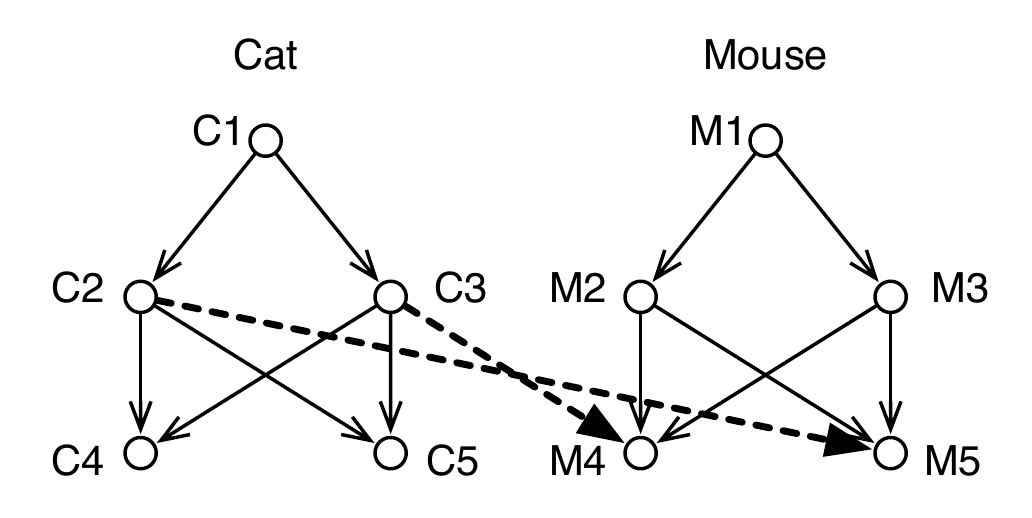}
\caption{The gate gadget with threat edges (dashed lines) to enforce AND-gate semantics.}
\label{fig:and-basic}
\end{figure}

The encoding of AND gates is accomplished by the gate gadget combined with \emph{threat edges} as shown in Figure~\ref{fig:and-basic}. These edges are shown as dashed lines for clarity but are the same as any other edge in the graph. As before node 1 is the output and nodes 4 and 5 are the inputs to the gate.

Suppose at some point Mouse moves to M1, with Cat mirroring to C1. There are three cases:
\begin{description}
\itemsep 0pt
\parskip 0pt
\item[Both inputs are true.]
There are paths from M4 and M5 to true gates. Mouse picks node M2 or M3. Cat picks either C2 (to prevent Mouse from going to M5) or C3 (to prevent Mouse from going to M4). Mouse can move to whichever node is not threatened by Cat and safely exit the gadget.
\item[Exactly one input is true.]
Cat moves to threaten the path down the true branch. Mouse can only go forward along the false branch.
\item[Neither input is true.]
Mouse can only go forward down a false branch.
\end{description}

\subsection{Forcing Cat to mirror Mouse: escape routes}
\label{sec:escape-routes}

It may be that Cat can prevent Mouse from winning when the circuit is true by not mirroring Mouse but rather taking some other path through the graph that later intercepts Mouse's path. To prevent this we add an \emph{escape route} to each branch of each gadget. Each escape route has the same length as any other forward path to Hole. If Cat mirrors Mouse, Cat prevents Mouse from going down the escape route. Otherwise Mouse can safely get to Hole down the escape route. Formally, the construction is:
\begin{quote}
\em{Given a gadget, let $k$ be the length of any path from the bottom layer (C4, C5, M4, M5) to $h$. Create a chain of $k - 2$ new nodes $t_1, \ldots, t_{k - 2}$ with an edge from each node to the next. Add edges from C4 and M4 to $t_1$ and an edge from $t_{k - 2}$ to $h$. Similarly, create a (disjoint) route from C5 and M5 to $h$.}
\end{quote}

The routes are at the bottom of each gadget to handle a technicality: suppose there is a true AND gate where both inputs are the same node. Suppose Cat moves to C2, forcing Mouse to M4. Cat can cross to M5 and intercept Mouse at the next gate, which is connected to both M4 and M5. Placing the escape routes at M4 and M5 forces Cat to mirror Mouse exactly.

\subsection{Forcing Mouse forward (undirected case): guard edges}
\label{sec:guard-edges}

\begin{figure}[ht]
\centering
\includegraphics{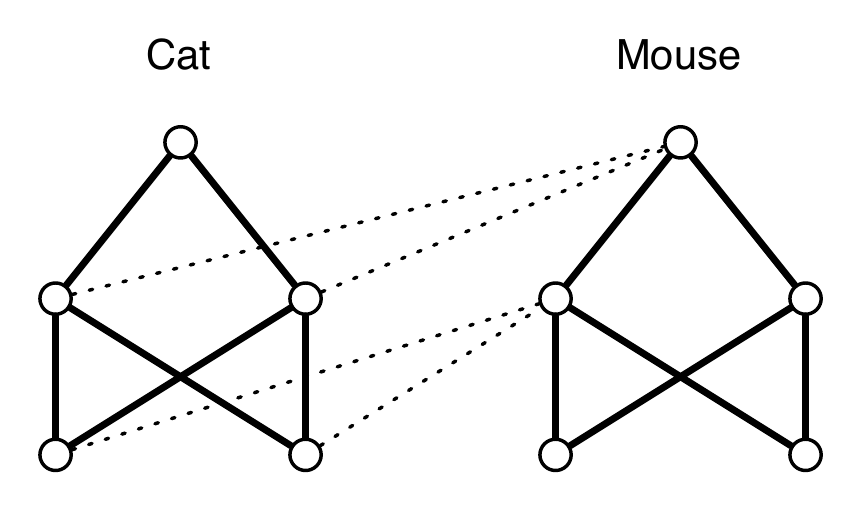}
\caption{A matched pair of gate gadgets with guard edges (dotted lines) for M1 and M2 shown (undirected case).}
\label{fig:guard-edges}
\end{figure}

The addition of \emph{guard edges} (see Figure~\ref{fig:guard-edges}) in step~\ref{item:guard-edges} of the undirected construction prevents Mouse from backtracking. The guard edges are shown with dotted lines for clarity; they are no different from any other edge.

Suppose that Mouse has moved to M2, with Cat mirroring to C2. If Mouse backtracks to M1, Cat can immediately catch Mouse by moving backward along the corresponding guard edge.

The rule for placing guard edges is:
\begin{quote}
\em{For any adjacent nodes $m_1$ and $m_2$ in $\gmouse$, with $m_2$ closer to Hole (so that Mouse backtracks if it moves from $m_2$ to $m_1$), place an edge between $m_1$ and $\gcat(m_2)$.}
\end{quote}
There is a one-to-one correspondence among guard edges, edges in $\gmouse$, and edges in $\gcat$. Thus in Figure~\ref{fig:guard-edges} there would also be guard edges going back from C1, from M3 to C4 and C5, and forward from nodes M4 and M5.

\section{Proof of the construction}

In the following, we show that given $C$ and $x$, and the constructions of $G_D$ and $G_U$, that:
\begin{align*}
\langle C, x \rangle &\in \smcvp{} \Leftrightarrow \\
\langle G_D, c, m, h \rangle &\in \dhc{} \Leftrightarrow \\
\langle G_U, c, m, h \rangle &\in \uhc{}.
\end{align*}

Recall that every edge in $G_D$ goes between two adjacent layers, and therefore that all paths from a given layer to the Hole have the same length. The same holds for all monotonic paths from a given layer in $G_U$.

\subsection{Directed case}

\begin{claim}
\label{claim:mouse-wins-directed}
Mouse wins if circuit is true.
\end{claim}
\begin{proof}
Mouse's strategy is to take a path of true gates from the output through a true input to the hole. We consider all possible strategies for Cat.
\begin{description}
\itemsep 0pt
\parskip 0pt
\item[Cat mirrors Mouse.]
Since the circuit is true, the top gate is true. If it is an OR gate then Mouse can advance along one of the branches into another true gate. If it is an AND gate then regardless of which branch Cat threatens, Mouse can avoid capture and advance to a true gate. Continuing in this way, since Mouse only moves through true gates Mouse will eventually advance to a true literal and thus to the Hole.
\item[Cat takes a non-mirroring path in $\gcat$.]
Mouse can get to $h$ along an escape route.
\item[Cat moves along a threat edge into $\gmouse$.]
Mouse will be on a branch with an escape route that is not threatened by Cat. Mouse can safely get to $h$ along this route.
\item[Cat moves onto an escape route.]
Mouse can freely get to $h$ along another escape route.
\end{description}
\end{proof}

\begin{claim}
\label{claim:cat-wins-directed}
Cat wins if circuit is false.
\end{claim}
\begin{proof}
Cat's strategy is to mirror Mouse until Mouse can be captured. We consider all possible strategies of Mouse:
\begin{description}
\itemsep 0pt
\parskip 0pt
\item[Mouse takes a path in $\gmouse$.]
Since the circuit is false, the top gate (where Mouse starts) is false. If it is an OR gate then Mouse can only advance into a false gate. If it is an AND gate then Cat can threaten the true branch (if any) and force Mouse to advance to a false gate. Continuing in this way, since Mouse moves only through false gates Mouse must eventually advance to a false literal and then to the dead end to be captured (see item~\ref{item:dead-end} of the construction).
\item[Mouse moves onto an escape route.]
By the construction of the escape routes, since Cat is mirroring, Cat is adjacent to the node Mouse has just occupied and can capture Mouse.
\end{description}
\end{proof}

\subsection{Undirected case}

The undirected case adds two types of moves that were not possible before: moving (forward or backward) along a guard edge and moving backward along any other edge. We consider these in turn:
\begin{enumerate}
\itemsep 0pt
\parskip 0pt
\item \textit{If Cat mirrors Mouse and Mouse backtracks within $\gmouse$ then Cat wins.}
Suppose Mouse backtracks from $v$ to $u$. Since Cat has been mirroring Mouse Cat is at $\gcat(v)$ and there is a guard edge from $\gcat(v)$ to $u$ by which Cat captures Mouse.
\item \textit{If Cat mirrors Mouse and Mouse backtracks along a threat edge into $\gcat$ then Cat wins.}
This only happens if Mouse moves from M4 or M5 of a gadget to C3 or C2. By mirroring, Cat is at C4 or C5. Cat can capture Mouse by moving back along a diagonal edge.
\item \textit{If Cat mirrors Mouse and Mouse moves forward along a guard edge into $\gcat$ then Cat wins.}
Since Cat has been mirroring Mouse, there is an edge in $\gcat$ that corresponds to the guard edge that Cat can take to capture Mouse.
\item \textit{If Mouse moves toward $h$ in $\gmouse$ and Cat backtracks then Mouse wins.}
Because of escape routes there is always a forward path from any node to $h$ (except for the false literal nodes, which we can ignore for this case). If Cat backtracks then Mouse is one layer closer to $h$, with Mouse to move. Mouse can freely advance to $h$, and since all forward paths have the same length there is no other path Cat can take to catch up.
\end{enumerate}

It follows that none of the additional possibilities offered by the undirected case changes the outcome from that of the directed case.

\section*{Acknowledgments}
This work was supported by the National Science Foundation under grants CCF-0829754 and CCF-1255900. The author also thanks Lance Fortnow for suggesting this problem and for several helpful discussions.

\bibliographystyle{plain}
\bibliography{happy.bib}
\end{document}